\documentclass[submission,copyright,creativecommons]{eptcs}
\providecommand{\event}{WPTE 2017} 
\usepackage{multicol}
\usepackage{color}
\definecolor{Navy}{rgb}{0.0,0.20,0.60}

\DeclareMathAlphabet{\mathcal}{OMS}{cmsy}{m}{n}

\usepackage{theorem}

\newcommand{\eop}{\hspace*{\fill}$\Box$}
\def\qed{\eop}
\newtheorem{definition}{Definition}[section]
\newtheorem{lemma}[definition]{Lemma}
\newtheorem{theorem}[definition]{Theorem}
{\theorembodyfont{\rmfamily}
  \newtheorem{example}[definition]{\it Example}
  \newtheorem{proof}{\it Proof.}
}

\usepackage{amsmath}
\usepackage{amssymb}
\usepackage{stmaryrd}

\def\ttunderscore{\char95}

\def\cE{\mathcal{E}}
\def\cH{\mathcal{H}}
\def\cR{\mathcal{R}}
\def\cS{\mathcal{S}}
\def\cV{\mathcal{V}}
\newcommand{\funtype}{\Rightarrow}
\def\cI{\mathcal{I}}
\newcommand{\cJ}{\mathcal{J}}
\newcommand{\Interpret}[1]{\llbracket #1 \rrbracket}
\newcommand{\Bool}{\mathbb{B}}
\def\cC{\mathcal{C}}
\def\cD{\mathcal{D}}
\def\Constrained#1{[\,#1\,]}
\newcommand{\FVar}{{\mathcal{V}\mathit{ar}}}
\newcommand{\Dom}{{\mathcal{D}\mathit{om}}}
\newcommand{\cRcalc}{\cR_{\mathtt{calc}}}
\newcommand{\Sigmaterms}{\Sigma_{\mathit{terms}}}
\newcommand{\Sigmalogic}{\Sigma_{\mathit{theory}}}
\newcommand{\Sigmatheory}{\Sigma_{\mathit{theory}}}
\newcommand{\Sigmacore}{\Sigma_{\mathit{theory}}^{\mathit{core}}}
\newcommand{\Sigmaint}{\Sigma_{\mathit{theory}}^{\INT}}
\def\Val{{\mathcal{V}\mathit{al}}}
\newcommand{\LVar}{{\mathcal{LV}\mathit{ar}}}
\newcommand{\CTerm}[2]{#1\,\Constrained{#2}}
\def\sort#1{\mathit{#1}}
\newcommand{\symb}[1]{\mathsf{#1}}
\newcommand{\Int}{\mathbb{Z}}
\newcommand{\INT}{\mathit{int}}
\newcommand{\BOOL}{\mathit{bool}}
\def\Rule#1#2{#1 \to #2}
\def\CRule#1#2#3{#1 \to #2\ \Constrained{#3}}
\newcommand{\CEqn}[4][\approx]{#2 \mathrel{#1} #3\ \Constrained{#4}}
\newcommand{\Eqn}[3][\approx]{#2 \mathrel{#1} #3}
\def\Expd{\mathit{Expd}}
\def\RIstep{\vdash_{\mathit{RI}}}
\def\Expansion{\textsc{Expansion}}
\def\Simplification{\textsc{Simplification}}
\def\CaseSplitting{\textsc{CaseSplitting}}
\def\Deletion{\textsc{Deletion}}
\def\Generalization{\textsc{Generalization}}
\newcommand{\Emph}[2][blue]{\textcolor{#1}{#2}}
\def\cRsum{\cR_{\mathit{sum}}}
\def\cRsumcheck{\cR'_{\mathit{check}}}
\def\cRsumverif{\cR_1}
\def\cRcheck{\cR_{\mathit{check}}}
\newcommand{\lrang}[2][P]{\mathit{#1}} 
\newcommand{\true}{\symb{true}} 
\newcommand{\false}{\symb{false}} 
\newcommand{\Skip}{{\bf skip}}
\newcommand{\ifff}{{\bf if}} 
\newcommand{\elseee}{{\bf else}} 
\newcommand{\while}{{\bf while}} 

\newcommand{\triple}[3]{\left\lbrace #1\right\rbrace\> #2 \> \left\lbrace #3\right\rbrace}
\newcommand{\ttriple}[3]{[ #1 ] \> #2 \> [ #3 ]} 
\def\lnum#1{{\footnotesize #1}}
\def\anum#1{{\footnotesize \Emph[Navy]{#1}}}
\def\Lnum#1{\Emph[Navy]{#1}}
\def\Psum{P_{\mathit{sum}}}
\def\Psumtab{T_{\mathit{sum}}}
\def\Start{\mathit{start}}
\def\Check{\symb{chk}}
\def\State{\symb{state}}
\newcommand{\Eval}[1]{\Rightarrow_{#1}}
\def\TransOneStep{\mathit{Trans}_1}
\def\Trans{\mathit{Trans}}
\def\EP{e_P}

\title{Transforming Proof Tableaux of Hoare Logic into Inference Sequences of Rewriting Induction}

\def\titlerunning{Transforming Proof Tableaux of Hoare Logic into Inference Sequences of Rewriting Induction}
\author{Shinnosuke Mizutani
\institute{%
Graduate School of Information Science\\
Nagoya University\\
Nagoya, Japan
}
\email{mizutani{\ttunderscore}s@trs.cm.is.nagoya-u.ac.jp}
\and
Naoki Nishida
\institute{%
Graduate School of Informatics\\
Nagoya University\\
Nagoya, Japan
}
\email{nishida@i.nagoya-u.ac.jp}
}
\def\authorrunning{S.~Mizutani and N.~Nishida}

\begin{document}
\maketitle

\begin{abstract}
A proof tableau of Hoare logic is an annotated program with pre- and post-conditions, which corresponds to an inference tree of Hoare logic. 
In this paper, we show that a proof tableau for partial correctness can be transformed into an inference sequence of rewriting induction for constrained rewriting.
We also show that the resulting sequence is a valid proof for an inductive theorem corresponding to the Hoare triple if the constrained rewriting system obtained from the program is terminating.
Such a valid proof with termination of the constrained rewriting system implies total correctness of the program w.r.t.\ the Hoare triple.
The transformation enables us to apply techniques for proving termination of constrained rewriting to proving total correctness of programs together with proof tableaux for partial correctness. 
\end{abstract}




\section{Introduction}
\label{sec:intro}

In the field of term rewriting, automated reasoning about \emph{inductive theorems} has been well investigated.
Here, an inductive theorem of a term rewriting system (TRS) is an equation that is inductively valid, i.e., all of its ground instances are theorems of the TRS.
As principles for proving inductive theorems, we cite \emph{inductionless induction}~\cite{Mus80,Hue82} and \emph{rewriting induction} (RI)~\cite{Red90}, both of which are called \emph{implicit induction principles}.
Frameworks based on the RI principle (RI frameworks, for short) consist of inference rules to prove that given equations are inductive theorems.
On the other hand, RI-based methods are procedures within RI frameworks to apply inference rules under specified strategies.
In recent years, various RI-based methods for \emph{constrained rewriting} (see, e.g., constrained TRSs~\cite{FNSKS08b,SNS11}, conditional and constrained TRSs~\cite{BJ08}, $\mathbb{Z}$-TRSs~\cite{FK12}, and logically constrained TRSs~\cite{KN13frocos}) have been developed~\cite{BJ08,SNSSK09,FK12,KN14aplas,FKN17tocl}.
Constrained systems
have built-in semantics for some function and predicate symbols
and have been used as a computation model of not only functional but also imperative programs~\cite{FK08,FGPSF09,FNSKS08b,FK09,Vid12,KN14aplas,FKN17tocl}.

For program verification, several techniques have been investigated in the literature, e.g., \emph{model checking}, \emph{Hoare logic}, etc.
On the other hand, constrained rewriting can be used as a computation model of some imperative programs (cf.~\cite{FKN17tocl}), and RI frameworks for constrained rewriting are tuned to verification of imperative programs, e.g. equivalence of two functions under the same specification.
Some RI frameworks succeed in proving equivalence of an imperative program and its functional specification such that a proof based on Hoare logic needs a loop invariant (cf.~\cite{FKN17tocl}).
From such experiences, we are interested in differences between RI frameworks and other verification methods. 

In this paper, we show that a \emph{proof tableau} of Hoare logic can be transformed into an inference sequence of rewriting induction for \emph{logically constrained TRSs} (LCTRSs).
Here, a proof tableau is an annotated \emph{while} program with pre- and post-conditions, which corresponds to an inference tree of Hoare logic.
We also show that the resulting inference sequence is a valid proof for an inductive theorem corresponding to the Hoare triple for the proof tableau if the LCTRS obtained from the program is terminating.

Given a \emph{while} program $P$ and a proof tableau $T_P$ of a Hoare triple $\triple{\varphi_P}{P}{\psi_P}$ for \emph{partial correctness}, we proceed as follows:
\begin{enumerate}
	\item 
We transform 
$P$ 
into an equivalent LCTRS $\cR_P$, and we prove termination of the LCTRS $\cR_P$.
	\item 
We prepare rewrite rules $\cRcheck$ to verify the post-condition $\psi_P$ in the proof tableau.
	\item 
We prepare a constrained equation $e_P$ corresponding to the Hoare triple $\triple{\varphi_P}{P}{\psi_P}$.
	\item 
Starting with the equation $e_P$, we transform the proof tableau into an inference sequence $(\{e_P\},\emptyset)$ $\mathrel{\RIstep} \cdots \mathrel{\RIstep} (\emptyset,\cH)$ of RI in a top-down fashion, where we do not prove termination in constructing the inference sequence of RI.
\end{enumerate}
In addition to the above transformation, we show that termination of the LCTRS $\cR_P$ implies termination of the LCTRS $\cR_P \cup \cRcheck \cup \cH$.
Termination of the LCTRS $\cR_P\cup\cRcheck\cup\cH$ ensures that the resulting inference sequence $(\{e_P\},\emptyset) \mathrel{\RIstep} \cdots \mathrel{\RIstep} (\emptyset,\cH)$ is a valid proof of RI---the equation $e_P$ is an inductive theorem of the LCTRS $\cR_P$---and thus, the \emph{while} program $P$ is \emph{totally correct} w.r.t.\ $\varphi_P$ and $\psi_P$. 

The contribution of this paper is 
a top-down transformation of proof tableaux for partial correctness to inference sequences of RI, which enables us to apply techniques for proving termination of constrained rewriting to proving total correctness together with proof tableaux for partial correctness. 

This paper is organized as the follows.
In Section~\ref{sec:preliminaries}, we briefly recall LCTRSs, \emph{while} programs, and a conversion of \emph{while} programs to LCTRSs.
In Section~\ref{sec: hoare}, we recall proof tableaux of Hoare logic, and in Section~\ref{sec: ri}, we recall the framework of rewriting induction for LCTRSs.
In Section~\ref{sec: transform}, we show that a proof tableau can be transformed into an inference sequence of RI, and the resulting inference sequence is a valid proof for total correctness if the LCTRS obtained from the proof tableau is terminating.
In Section~\ref{sec:conclusion}, we conclude this paper and describe future direction of this research. 

\section{Preliminaries}
\label{sec:preliminaries}

In this section, we recall 
LCTRSs, following the definitions in~\cite{KN13frocos,FKN17tocl}.
We also recall \emph{while programs}, and then introduce a conversion of \emph{while} programs to LCTRSs. 
Familiarity with basic notions on term rewriting~\cite{BN98,Ohl02} is assumed.

\subsection{Logically Constrained Term Rewriting Systems}


Let $\cS$ be a set of \emph{sorts} and $\cV$ a countably infinite set of \emph{variables}, each of which is equipped with a sort.
A \emph{signature} $\Sigma$ is a set, disjoint from $\cV$, of \emph{function symbols} $f$, each of which is equipped with a \emph{sort declaration} $\iota_1 \times \cdots \times \iota_n \funtype \iota$ where $\iota_1,\ldots,\iota_n,\iota \in \cS$.
For readability, we often write $\iota$ instead of
$\iota_1 \times \cdots \times \iota_n  \funtype \iota$ if $n=0$.
We denote the set of well-sorted \emph{terms}  
over $\Sigma$ and $\cV$
by $T(\Sigma,\cV)$.
In the rest of this section, we fix $\cS$, $\Sigma$, and $\cV$.
The set of variables occurring in $s$ is denoted by $\FVar(s)$.
Given a term $s$ and a \emph{position} $p$ (a sequence of positive integers) of $s$, $s|_p$ denotes the subterm of $s$ at position $p$,
and
$s[t]_p$ denotes $s$ with the subterm at position $p$ replaced by $t$.

A \emph{substitution} $\gamma$ is a sort-preserving total mapping from $\cV$ to $T(\Sigma,\cV)$, and naturally extended for a mapping from $T(\Sigma,\cV)$ to $T(\Sigma,\cV)$:
the result $s\gamma$ of applying a substitution $\gamma$ to a term $s$ is $s$ with all occurrences of a variable $x$ replaced by $\gamma(x)$.
The \emph{domain} $\Dom(\gamma)$ of $\gamma$ is the set of variables $x$ with $\gamma(x) \neq x$.
The notation $\{ x_1 \mapsto s_1, \ldots, x_k \mapsto s_k \}$ denotes a substitution $\gamma$ with $\gamma(x_i) = s_i$ for $1 \leq i \leq n$, and $\gamma(y) = y$ for $y \notin \{x_1, \ldots,x_n\}$.
%


To define LCTRSs, we consider different kinds of symbols and terms:
(1) two signatures $\Sigmaterms$ and $\Sigmalogic$ such that $\Sigma = \Sigmaterms \cup \Sigmalogic$,
(2) a mapping $\cI$ which assigns to each sort $\iota$ occurring in $\Sigmatheory$ a set $\cI_\iota$,
(3) a mapping $\cJ$ which assigns to each $f : \iota_1 \times \cdots \times \iota_n \funtype \iota \in \Sigmalogic$ a function in $\cI_{\iota_1} \times \cdots \times \cI_{\iota_n} \funtype \cI_\iota$,
and
(4) a set $\Val_\iota \subseteq \Sigmalogic$ of \emph{values}---function symbols $a : \iota$ such that $\cJ$ gives a bijective mapping from $\Val_\iota$ to $\cI_\iota$---for each sort $\iota$ occurring in $\Sigmatheory$.
We require that $\Sigmaterms \cap \Sigmalogic \subseteq \Val = \bigcup_{\iota \in \cS} \Val_\iota$.
The sorts occurring in $\Sigmalogic$ are called \emph{theory sorts}, and the symbols \emph{theory symbols}.
Symbols in ${\Sigmalogic} \setminus {\Val}$ are \emph{calculation symbols}.
A term in $T(\Sigmalogic,\cV)$ is called a \emph{logical term}.
For ground logical terms, we define the interpretation as $\Interpret{f(s_1,\ldots,s_n)} = \cJ(f)(\Interpret{s_1},\ldots,\Interpret{s_n})$.
For every ground logical term $s$, there is a unique value $c$ such that $\Interpret{s} = \Interpret{c}$. 
We use infix notation for theory and calculation symbols.

A \emph{constraint} is a logical term $\varphi$ of some sort $\BOOL$ with $\cI_\BOOL = \Bool = \{ \top,\bot \}$, the set of \emph{booleans}.
A constraint $\varphi$ is \emph{valid} if $\Interpret{\varphi\gamma} = \top$ for all substitutions $\gamma$ which map $\FVar(\varphi)$ to values, and \emph{satisfiable} if $\Interpret{\varphi\gamma} = \top$ for some such substitution.
A substitution $\gamma$ \emph{respects} $\varphi$ if $\gamma(x)$ is a value for all $x \in \FVar(\varphi)$ and $\Interpret{\varphi\gamma} = \top$.
We typically choose a theory signature with $\Sigmalogic \supseteq \Sigmacore$, where $\Sigmacore$ contains
$\symb{true},\symb{false} : \BOOL$, $\wedge, \vee, \implies : \BOOL \times \BOOL \funtype \BOOL$, $\neg: \BOOL \funtype \BOOL$, and, for all theory sorts $\iota$, symbols $=_\iota, \neq_\iota : \iota \times \iota \funtype \BOOL$, and an evaluation function $\cJ$ that interprets these symbols as expected. 
We omit the sort subscripts from $=$ and $\neq$ when clear from context.

The standard integer signature $\Sigmaint$ is $\Sigmacore \cup \{ +, -,*,\symb{exp},\symb{div}, \symb{mod} : \INT \times \INT \funtype \INT \}\cup \{ {\geq}, {>} : \INT \times \INT \funtype \BOOL \} \cup \{ \symb{n} : \INT \mid n \in \Int \}$ with values $\symb{true}$, $\symb{false}$, and $\symb{n}$ for all integers $n \in \Int$.
\nopagebreak
Thus, we use $\symb{n}$ (in \textsf{sans-serif} font) as the function symbol for $n \in \Int$ (in $\mathit{math}$ font).
We define $\cJ$ in the natural way, except: since all $\cJ(f)$ must be total functions, we set $\cJ(\symb{div})(n,0) = \cJ(\symb{mod})(n,0) = \cJ(\symb{exp})(n,k) = 0$ for all $n$ and all $k < 0$. 
When constructing LCTRSs from, e.g., \emph{while} programs, we can add explicit error checks for, e.g., ``division by zero'', to constraints (cf.~\cite{FKN17tocl}).



A \emph{constrained rewrite rule} is a triple $\CRule{\ell}{r}{\varphi}$ such that $\ell$ and $r$ are terms of the same sort, $\varphi$ is a constraint, and $\ell$ has the form $f(\ell_1,\dots,\ell_n)$ and contains at least one symbol in $\Sigmaterms \setminus \Sigmatheory$ (i.e., $\ell$ is not a logical term).
If $\varphi = \symb{true}$ with $\cJ(\symb{true}) = \top$, we may write $\Rule{\ell}{r}$.
We define $\LVar(\CRule{\ell}{r}{\varphi})$ as $\FVar(\varphi) \cup (\FVar(r) \setminus \FVar(\ell))$.
We say that a substitution $\gamma$ \emph{respects} $\CRule{\ell}{r}{\varphi}$ if $\gamma(x) \in \Val$ for all $x \in \LVar(\CRule{\ell}{r}{\varphi})$, and $\Interpret{\varphi\gamma} = \top$.
Note that it is allowed to have $\FVar(r) \not\subseteq \FVar(\ell)$, but fresh variables in the right-hand side may only be instantiated with \emph{values}.
%
Given a set $\cR$ of constrained rewrite rules, we let $\cRcalc$ be the set $\{ \CRule{f(x_1,\ldots,x_n)}{y}{y = f(x_1,\ldots,x_n)} \mid f : \iota_1 \times \cdots \times \iota_n \funtype \iota \in {\Sigmalogic} \setminus {\Val} \}$.
We usually call the elements of $\cRcalc$ constrained rewrite rules (or \emph{calculation rules}) even though their left-hand side is a logical term.
The \emph{rewrite relation} $\to_{\cR}$ is a binary relation on terms, defined by:
$s[\ell\gamma]_p \mathrel{\to_\cR} s[r\gamma]_p$ if 
$\CRule{\ell}{r}{\varphi} \in \cR \cup \cRcalc$ and $\gamma$ respects $\CRule{\ell}{r}{\varphi}$.
A reduction step with $\cRcalc$ is called a \emph{calculation}.

Now we define a \emph{logically constrained term rewriting system} (LCTRS) as the abstract rewriting system $(T(\Sigma,\cV),\to_\cR)$.
An LCTRS is usually given by supplying $\Sigma$, $\cR$, and an informal description of $\cI$ and $\cJ$ if these are not clear from context.
An LCTRS $\cR$ is said to be \emph{left-linear} if for every rule in $\cR$, the left-hand side is linear.
$\cR$ is said to be \emph{non-overlapping} if
  for every term $s$ and rule $\CRule{\ell}{r}{\varphi}$ such that $s$ reduces with $\CRule{\ell}{r}{\varphi}$ at the root
  position: (a) there are no other rules $\CRule{\ell'}{r'}{\varphi'}$ such that $s$ reduces
  with $\CRule{\ell'}{r'}{\varphi'}$ at the root position, and (b) if $s$ reduces with any
  rule at a non-root position $q$, then $q$ is not a position of
  $\ell$.
$\cR$ is said to be \emph{orthogonal} if $\cR$ is left-linear and non-overlapping.
For $\CRule{f(\ell_1,\ldots,\ell_n)}{r}{\varphi} \in \cR$, we call $f$ a \emph{defined symbol} of $\cR$, 
and non-defined elements of $\Sigmaterms$ and all values are called \emph{constructors} of $\cR$.
Let $\cD_\cR$ be the set of all defined symbols and $\cC_\cR$ the set of constructors.
A term in $T(\cC_\cR,\cV)$ is a \emph{constructor term} of $\cR$.

\begin{example}[\cite{FKN17tocl}]
Let $\cS = \{ \INT,\BOOL \}$, and
$\Sigma = \Sigmaterms \cup \Sigmaint$, where
$
\Sigmaterms = \{~ \symb{fact} : \INT \funtype \INT ~\} \cup \{~ \symb{n} : \INT \mid n \in \Int ~\}
$.
Then both $\INT$ and $\BOOL$ are theory sorts.
We also define set and function interpretations, i.e., $\cI_\INT = \Int$, $\cI_\BOOL = \Bool$, and $\cJ$ is defined as above. 
Examples of logical terms are $\symb{0} = \symb{0}+\symb{-1}$ and
$x+\symb{3} \geq y + -\symb{42}$ that are constraints.
$\symb{5}+\symb{9}$ is also a (ground) logical term, but not a constraint.
Expected starting terms are, e.g.,
$\symb{fact}(\symb{42})$ or $\symb{fact}(\symb{fact}(\symb{-4}))$.
\label{exa:factlctrs}
To implement an LCTRS calculating the \emph{factorial} function, we use the signature $\Sigma$ above 
 and the following rules:
$
\cR_{\symb{fact}}
 = \{\ 
\CRule{\symb{fact}(x)}{\symb{1}}{x \leq \symb{0}},
~~
\CRule{\symb{fact}(x)}{x \times \symb{fact}(x-\symb{1})}{\neg (x \leq \symb{0})}
\ \}
$. 
Using calculation steps, a term $\symb{3}-\symb{1}$ reduces to $\symb{2}$ in one step with the calculation rule $\CRule{x-y}{z}{z = x-y}$, and $\symb{3} \times (\symb{2} \times (\symb{1} \times \symb{1}))$ reduces to $\symb{6}$ in three steps. 
Using the constrained rewrite rules in $\cR_{\symb{fact}}$, $\symb{fact}(\symb{3})$ reduces in ten steps to $\symb{6}$.
\end{example}

A \emph{constrained term} is a pair $\CTerm{s}{\varphi}$ of a term $s$ and a constraint $\varphi$.
We say that $\CTerm{s}{\varphi}$ and $\CTerm{t}{\psi}$ are \emph{equivalent}, written by $\CTerm{s}{\varphi} \sim \CTerm{t}{\psi}$, if for all substitutions $\gamma$ which respect $\varphi$, there is a substitution $\delta$ which respects $\psi$ such that $s\gamma = t\delta$, and vice versa.
Intuitively, a constrained term $\CTerm{s}{\varphi}$ represents all terms $s\gamma$ where $\gamma$ respects $\varphi$, and can be used to reason about such terms.
For this reason, equivalent constrained terms represent the same set of terms.
%
For a rule $\rho := \CRule{\ell}{r}{\psi}\in \cR \cup \cRcalc$ and position $q$, we let $\CTerm{s}{\varphi} \mathrel{\to_{\rho,q}} \CTerm{t}{\varphi}$ if there exists a substitution $\gamma$ such that $s|_q = \ell\gamma$, $t = s[r\gamma]_q$, $\gamma(x)$ is either a value or a variable in $\FVar(\varphi)$ for all $x \in \LVar(\CRule{\ell}{r}{\psi})$, and $\varphi \implies (\psi\gamma)$ is valid.
We write $\CTerm{s}{\varphi} \mathrel{\to_{\mathtt{base}}} \CTerm{t}{\varphi}$ for $\CTerm{s}{\varphi} \mathrel{\to_{\rho,q}} \CTerm{t}{\varphi}$ with some $\rho,q$.
The relation $\to_\cR$ on constrained terms is defined as ${\sim} \cdot {\to_{\mathtt{base}}} \cdot {\sim}$. 

\subsection{While Programs}

In this section, we recall the syntax of \emph{while programs} (see e.g.,~\cite{Rey98}).

We deal with a simple class of \emph{while} programs over the integers, which consist of assignments, skip, sequences, ``if'' statements, and ``while'' statements with loop invariants:
a ``while'' statement is of the form $\while\, @\,\zeta\,(\psi) \{ c \}$ with $\zeta$ a loop invariant.
To deal with proof tableaux, we allow to write assertions of the form $@ \varphi$ as annotations.
An \emph{annotated} \emph{while program} is defined by the following BNF:
\begin{eqnarray*}
\lrang{comm} &::=& v := \lrang[E]{intexp} \mid \Skip \mid  \lrang{comm}\,; \lrang{comm} \mid @\,\lrang[B]{bexp} 
                      \mid \ifff ( \lrang[B]{bexp} ) \{ \lrang{comm} \} \elseee \{ \lrang{comm} \}
                          \mid \while\, @\,\lrang[B]{bexp}\,(\lrang[B]{bexp}) \{ \lrang{comm} \} \\
\lrang[E]{intexp} &::=& n 
                      \mid v 
						\mid (\lrang[E]{intexp} \mathrel{+} \lrang[E]{intexp}) 
						\mid (\lrang[E]{intexp} \mathrel{-} \lrang[E]{intexp}) 
						\mid (\lrang[E]{intexp} \mathrel{*} \lrang[E]{intexp}) 
						\mid (\lrang[E]{intexp} \mathrel{/} \lrang[E]{intexp}) 
						\\
\lrang[B]{bexp} &::=& \true \mid \false
                         \mid \lrang[E]{intexp} = \lrang[E]{intexp} 
                         \mid \lrang[E]{intexp} < \lrang[E]{intexp} 
                         \mid (\lnot \lrang[B]{bexp})
                         \mid (\lrang[B]{bexp} \lor \lrang[B]{bexp}) 
\end{eqnarray*}
where $n \in \Int$, $v \in \cV$, and we may omit brackets in the usual way.
We use $\ne$, $\leq$, $>$, $\geq$, $\land$, $\implies$, etc, as syntactic sugars.
We abbreviate $\while\, @\,\true\,(\psi) \{ c \}$ to $\while(\psi)\{ c \}$.
For page limitation, we do not introduce the semantics of \emph{while} programs, and they are evaluated in the usual way:
in evaluating \emph{while} programs, we ignore loop invariants and assertions, while they are taken into account in considering proof tableaux. 
For a \emph{while} program $P$, we denote the set of variables appearing in $P$ by $\FVar(P)$.
Given an assignment $\theta$ for $\FVar(P)$, we write $\theta \mathrel{\Eval{P}} \theta'$ if the execution of $P$ starts with $\theta$ and halts with an assignment $\theta'$.
We abuse assignments for variables as substitutions for terms.

\begin{example}
\label{ex:sum}
The following, denoted by $\Psum$, is a \emph{while} program with 
$\FVar(\Psum)=\{x,i,z\}$,
which computes the summation from $0$ to $x$ if $x \geq 0$.
\begin{tabbing}
00000000 \=00 \= 00 \= 00 \= 00 \= 00 \kill
\> \lnum{1}\' \> $i := 0;$\\
\> \lnum{2}\' \> $z := 0;$\\
\> \lnum{3}\' \> $\while (x > i)\{$\\
\> \lnum{4}\' \> \> $z := z + i + 1;$\\
\> \lnum{5}\' \> \> $i := i + 1;$\\
\> \lnum{6}\' \> $\}$\\
\> \lnum{7}\' 
\end{tabbing}
We write a line number for each statement, and write a blank line at the end of the program, which is used to simplify a conversion of \emph{while} programs to LCTRSs.
\end{example}

\subsection{Converting \emph{while} Programs to LCTRSs}
\label{subsec:while_to_lctrs}

In this section, we briefly introduce a conversion of \emph{while} programs to LCTRSs (see e.g.,~\cite{FKN17tocl}).

Let $P$ be a \emph{while} program such that $\FVar(P)=\{x_1,\ldots,x_n\}$ and $P$ has $m$ lines without any assertion.
We denote the sequence ``$x_1,\ldots,x_n$'' by $\vec{x}$.
We prepare $\sort{state}$, a sort for tuples of integers.
We assume that there is no blank line in $P$ with line numbers, except for the last line $m$ e.g., line 7 of $\Psum$.
We first prepare $m$ function symbols $\State_1, \ldots, \State_m$ with sort $\overbrace{\Int \times \cdots \times \Int}^n \funtype \sort{state}$.
Instances of $\State_1,\ldots,\State_m$ represent \emph{states} in executing $P$.
Here, a state consists of a program counter and an assignment to variables in the program (see e.g.,~\cite{BM07}).
For example, $\State_i(v_1,\ldots,v_n)$ represents a state such that the program counter stores $i$ and $v_1,\ldots,v_n$ are assigned to $x_1,\ldots,x_n$, resp.
For each statement in $P$, 
we generate constrained rewrite rules for $\State_1, \ldots, \State_m$ as follows:
\begin{itemize}
	\item an assignment \,\fbox{~\lnum{$i$}~~~~$x_k := e$;~} \,
	is converted to the following rule:
	\[
	\{~~
	\Rule{\State_i(\vec{x})}{\State_{i+1}(x_1,\ldots,x_{k-1},e,x_{k+1},\ldots,x_n)}
	~~\}
	\]
	\item a ``skip'' statement \,\fbox{\footnotesize~\lnum{$i$}~~~~$\Skip$;~}\,
	is converted to the following rule:
	\[
	\{~~
	\Rule{\State_i(\vec{x})}{\State_{i+1}(\vec{x})}
	~~\}
	\]
	\item an ``if'' statement
\,\fbox{%
\footnotesize
\begin{minipage}[t]{30ex}
\begin{tabbing}
00\=00 \= 00 \= 00 \= 00 \= 00 \kill
\> \lnum{$i$}\' \> $\ifff (\varphi)\{$~~~~~~~~~~\\
\> \lnum{\vdots}\' \>\> \raisebox{3pt}{$\cdots$} \\
\> \lnum{$j$}\' \> $\}\elseee\{$ \\
\> \lnum{\vdots}\' \>\> \raisebox{3pt}{$\cdots$} \\
\> \lnum{$k$}\' \>$\}$ 
\end{tabbing}
\end{minipage}
}\,
	is converted to the following rules:
	\[
\left\{
\begin{array}{r@{\>}c@{\>}ll@{}c@{}l@{~~~~~~~~~~~~~}r@{\>}c@{\>}ll@{}c@{}l}
\CRule{\State_i(\vec{x}) &}{& \State_{i+1}(\vec{x}) &}{& \varphi &} &
\Rule{\State_{j}(\vec{x}) &}{& \State_{k+1}(\vec{x})} \\
\CRule{\State_i(\vec{x}) &}{& \State_{j+1}(\vec{x}) &}{& \neg \varphi &} &
\Rule{\State_{k}(\vec{x}) &}{& \State_{k+1}(\vec{x})} \\
\end{array}
\right\}
	\]
	\item a ``while'' statement
\,\fbox{%
\footnotesize
\begin{minipage}[t]{30ex}
\begin{tabbing}
00\=00 \= 00 \= 00 \= 00 \= 00 \kill
\> \lnum{$i$}\' \> $\while\,@\zeta\, (\varphi)\{$~~~~~~~~~~\\
\> \lnum{\vdots}\' \>\> \raisebox{3pt}{$\cdots$} \\
\> \lnum{$j$}\' \> $\}$ 
\end{tabbing}
\end{minipage}
}\,
	is converted to the following rules:
	\[
\left\{
\begin{array}{r@{\>}c@{\>}ll@{}c@{}l@{~~~~~~~~~~~~~}r@{\>}c@{\>}ll@{}c@{}l}
\CRule{\State_i(\vec{x}) &}{& \State_{i+1}(\vec{x}) &}{& \varphi &} &
\Rule{\State_{j}(\vec{x}) &}{& \State_{i}(\vec{x})} \\
\CRule{\State_i(\vec{x}) &}{& \State_{j+1}(\vec{x}) &}{& \neg \varphi &} \\
\end{array}
\right\}
	\]
\end{itemize}
For brevity, we replace $\State_m$ in the final result by $\symb{end}$.
By definition, it is clear that any LCTRS obtained from a \emph{while} program by the above conversion is orthogonal.

\begin{example}
The program $\Psum$ 
 in Example~\ref{ex:sum} 
is converted to 
the following LCTRS:
\[
	\cRsum 
	=
	\left\{
	\begin{array}{r@{\>}c@{\>}ll@{}c@{}l}
\Rule{\State_1(x,i,z) &}{& \State_2(x,\symb{0},z)} \\
\Rule{\State_2(x,i,z) &}{& \State_3(x,i,\symb{0})} \\
\CRule{\State_3(x,i,z) &}{& \State_4(x,i,z) &}{& x > i &} \\
\CRule{\State_3(x,i,z) &}{& \symb{end}(x,i,z) &}{& \neg (x > i) &}\\
\Rule{\State_4(x,i,z) &}{& \State_5(x,i,z+i+\symb{1})} \\
\Rule{\State_5(x,i,z) &}{& \State_6(x,i+\symb{1},z)} \\
\Rule{\State_6(x,i,z) &}{& \State_3(x,i,z)}\\
	\end{array}
	\right\}
\]
$\cRsum$ is orthogonal (and thus, confluent), \emph{quasi-reductive} (i.e., every ground term with a defined symbol is reducible), and terminating.
Note that termination of $\cRsum$ can be proved by e.g., \textsf{Ctrl}~\cite{KN15lpar}.
\end{example}

\begin{theorem}[\cite{FNSKS08b}]
\label{thm:correctness_of_conversion}
Let $\cR_P$ be the LCTRS obtained from $P$ by the conversion in this section.
For all assignments $\theta,\theta'$ (for $\FVar(P)$),
$\theta \mathrel{\Eval{P}} \theta'$ iff $\State_1(\vec{x})\theta \mathrel{\to^*_{\cR_P}} \symb{end}(\vec{x})\theta'$.	
\end{theorem}
Note that the execution of $P$ starting with $\theta$ does not halt iff $\State_1(\vec{x})\theta$ does not terminate on $\cR_P$.
It follows from Theorem~\ref{thm:correctness_of_conversion} that if $\cR_P$ is terminating, then any execution of $P$ halts.
On the other hand, the converse does not hold for all \emph{while} programs, i.e., the conversion above does not preserve termination of $P$ (see, e.g.,~\cite{NK14wst}).%
\footnote{
When replacing $x > i$ in $\Psum$ by $x \ne i$, the constructed LCTRS $\cRsum'$  is the one obtained from $\cRsum$ by replacing $x > i$ by $x \ne i$.
LCTRS $\cRsum'$ is not terminating because we have an infinite reduction sequence from, e.g., $\State_3(\symb{0},\symb{1},\symb{0})$.
}

\section{Proof Tableaux of Hoare Logic}
\label{sec: hoare}

\emph{Hoare logic} is a logic to prove a Hoare triple
to hold (see e.g.,~\cite{Rey98}).
A triple $\triple{\varphi}{P}{\psi}$ for partial correctness is said to \emph{hold} (or $P$ is \emph{partially correct} w.r.t.\ pre- and post-conditions $\varphi$, $\psi$) if for any initial state satisfying $\varphi$, the final state of the execution satisfies $\psi$ whenever the execution from the initial state halts.
A triple $\ttriple{\varphi}{P}{\psi}$ for total correctness is said to \emph{hold} (or $P$ is \emph{totally correct} w.r.t.\ pre- and post-conditions $\varphi$, $\psi$) if for any initial state satisfying $\varphi$, the execution from the initial state halts and the final state of the execution satisfies $\psi$.
Note that total correctness is equivalent to partial correctness with termination of the program under the pre-condition.

In this section, we formalize proof tableaux of Hoare triples.
The aim of this paper is to transform a proof tableau of a Hoare triple for partial correctness into an inference sequence of RI (shown in Section~\ref{sec: ri}).
For this reason, we consider proof tableaux for partial correctness and we do not focus on the construction of proof tableaux.

In the following, we consider \emph{while} programs as sequences of commands connected by ``;'', and we write $P$ as $C_1;C_2;\ldots;C_n$.
Note that we consider ``;'' to implicitly exist at the end of ``if'' and ``while'' statements.
Bodies of ``if'' and ``while'' statements are also considered sequences of commands. 

\begin{definition}
An annotated \emph{while} program $P$ is called a \emph{proof tableau} if all of the following hold:
\begin{itemize}
	\item every longest command-(sub)sequence in $P$ has the length more than two, and the head and last elements of the sequence are annotations,
		e.g., 
		$P$ is of the form $@\,\varphi;C_1;\ldots;C_n;@\,\psi$ ($n > 0$),
	\item for each subsequence $@\,\varphi;@\,\psi$ of annotations, the formula $\varphi \implies \psi$ is valid, and
	\item for each subsequence $C_1;C_2;C_3$, if $C_2$ is not an annotation, then the first and third elements $C_1$, $C_3$ are annotations such that
	\begin{itemize}
		\item if\/ $C_2$ is an assignment $x := e$, then $C_1$ is $C_3\{x \mapsto e\}$, 
		\item if\/ $C_2$ is $\Skip$, then $C_1$ and $C_3$ are equivalent,
		\item if\/ $C_2$ is of the form $\ifff ( \psi ) \{ S' \} \elseee \{ S'' \}$ and $C_1$ is of the form $@\,\varphi$, then the head of $S'$ is $@\,\varphi \wedge \psi$, the head of $S''$ is $@\,\varphi \wedge \neg \psi$, and $C_3$ and the last elements of both $S'$ and $S''$ are equivalent, i.e., $C_1;C_2;C_3$ is of the form
			\[
			@\varphi;\, \ifff ( \psi ) \{ @\varphi\wedge\psi;\ldots;\,@\xi \} \elseee \{ @\varphi\wedge\neg\psi;\ldots;\,@\xi \};\,@\xi
			\]
			and
		\item if\/ $C_2$ is of the form $\while\,@ \zeta\,( \varphi ) \{ S\}$, then $C_1$ and the last element of the sequence $S$ are $@\zeta$, and the head element of $S$ is $@\,\zeta \wedge \varphi$, and $C_3$ is $@\, \zeta \wedge \neg \varphi$, i.e., $C_1;C_2;C_3$ is of the form 
			\[
			@\zeta;\,\while\,@\zeta\,( \varphi ) \{ @\zeta\wedge\varphi;\ldots;\,@\zeta\};\,@\zeta\wedge\neg\varphi.
			\]
	\end{itemize}
\end{itemize}
\end{definition}
Note that a proof tableau is a tableau representation of an inference tree constructed by basic inference rules of Hoare logic illustrated in Figure~\ref{fig:Hoare} (see e.g.,~\cite{Rey98}).
\begin{figure}[tb]
\[
\frac{~\mbox{$\varphi \implies \varphi'$ is valid}~~~~\triple{\varphi'}{C}{\psi'}~~~~ \mbox{$\psi' \implies \psi$ is valid}~}{~\triple{\varphi}{C}{\psi}~}
~~~~
\frac{}{~\triple{\varphi\{v \mapsto e\}}{v := e}{\varphi}~}
\]
\[
\frac{}{~\triple{\varphi}{\Skip}{\varphi}~}
~~~~
\frac{~\triple{\varphi}{C_1}{\xi}~~~~\triple{\xi}{C_2}{\psi}~}{~\triple{\varphi}{C_1;\,C_2}{\psi}~}
\]
\[
\frac{~\triple{\varphi \land \psi}{C_1}{\xi}~~~~\triple{\varphi \land \neg \psi}{C_2}{\xi}~}{~\triple{\varphi}{\ifff ( \psi ) \{ C_1 \} \elseee \{ C_2 \} }{\xi}~}
~~~~
\frac{~\triple{\zeta \land \psi}{C}{\zeta}~}{~\triple{\zeta}{\while\,\mbox{$@\,\zeta$}\, ( \psi ) \{ C \} }{\zeta \land \neg \psi }~}
\]
\vspace{-12pt}
\caption{basic inference rules of Hoare logic.}
\label{fig:Hoare}
\end{figure}

\begin{example}
The annotated \emph{while} program of Figure~\ref{fig:Ptab}, denoted by $\Psumtab$, is a proof tableau for the Hoare triple $\triple{x \geq 0}{\Psum}{z = \frac{1}{2}x(x + 1)}$, where the original line numbers for $\Psum$ are left.
\end{example}

\begin{figure}[t]
\begin{multicols}{2}
\begin{tabbing}
0000 \=00 \= 00 \= 00 \kill
\> \anum{A1} \' \> \Emph[Navy]{@\, $x \geq 0$;} \\
\> \anum{A2} \'\> \Emph[Navy]{@\, $x \geq 0 \land 0 = 0$;} \\
\> \lnum{1} \' \> $i := 0;$\\
\> \anum{A3} \'\> \Emph[Navy]{@\, $x \geq 0 \land i = 0$;} \\
\> \anum{A4} \'\> \Emph[Navy]{@\, $x \geq 0 \land i = 0 \land 0 = 0$;}\\
\> \lnum{2} \' \> $z := 0;$\\
\> \anum{A5} \'\> \Emph[Navy]{@\, $x \geq 0 \land i = 0 \land z = 0$;}\\
\> \anum{A6} \'\> \Emph[Navy]{@\, $z = \frac{1}{2}i(i+1) \land x \geq i$;}\\
\> \lnum{3} \' \> $\while\,\Emph[Navy]{@ z = \frac{1}{2}i(i+1) \land x \geq i}\, ~  (x > i)\{$\\
\> \anum{A7} \'\> \> \Emph[Navy]{@\, $z = \frac{1}{2}i(i+1) \land x \geq i \land x > i$;}
\end{tabbing}
\columnbreak
\begin{tabbing}
000 \=00 \= 00 \= 00 \kill
\> \anum{A8} \'\> \> \Emph[Navy]{@\, $z + i + 1 = \frac{1}{2}(i + 1)(i + 2) \land x \geq i + 1$;} \\
\> \lnum{4} \' \> \> $z := z + i + 1;$\\
\> \anum{A9} \'\> \> \Emph[Navy]{@\, $z = \frac{1}{2}(i + 1)(i + 2) \land x \geq i + 1$;} \\
\> \lnum{5} \' \> \> $i := i + 1;$\\
\> \anum{A10} \'\> \> \Emph[Navy]{@\, $z = \frac{1}{2}i(i + 1) \land x \geq i$;} \\
\> \lnum{6} \' \> $\}$\\
\> \anum{A11} \'\> \Emph[Navy]{@\, $z = \frac{1}{2}i(i + 1) \land x \geq i \land \neg (x > i)$;} \\[1pt]
\> \anum{A12} \'\> \Emph[Navy]{@\, $z = \frac{1}{2}x(x + 1)$;} \\
\> \lnum{7} \'
\end{tabbing}
\end{multicols}
\vspace{-12pt}
\caption{an annotated \emph{while} program $\Psumtab$ for $\Psum$.}
\label{fig:Ptab}
\end{figure}

\section{Rewriting Induction on LCTRSs}
\label{sec: ri}

In this section, we recall the framework of \emph{rewriting induction} (RI) for LCTRSs~\cite{FKN17tocl}. 

A \emph{constrained equation} is a triple 
$\CEqn{s}{t}{\varphi}$. 
We may simply write $\Eqn{s}{t}$ instead of $\CEqn{s}{t}{\varphi}$ if $\varphi$ is $\symb{true}$.
We write $\CEqn[\simeq]{s}{t}{\varphi}$ to denote either $\CEqn{s}{t}{\varphi}$ or $\CEqn{t}{s}{\varphi}$.
A substitution $\gamma$ is said to \emph{respect} $\CEqn{s}{t}{\varphi}$ if $\gamma$ respects $\varphi$ and $\FVar(s) \cup \FVar(t) \subseteq \Dom(\gamma)$, and to be a \emph{ground constructor substitution} if all $\gamma(x)$ with $x\in \Dom(\gamma)$ are ground constructor terms.
An equation $\CEqn{s}{t}{\varphi}$ is called an \emph{inductive theorem} of an LCTRS $\cR$ if $s\gamma \mathrel{\leftrightarrow^*_\cR} t\gamma$ for any ground constructor substitution $\gamma$ that respects $\CEqn{s}{t}{\varphi}$.

As in~\cite{FKN17tocl}, we restrict LCTRSs to be terminating and quasi-reductive.
An RI-based method is to construct an \emph{inference sequence} by applying the following basic inference rules to pairs $(\cE,\cH)$ of finite sets $\cE$ and $\cH$ of constrained equations and rewrite rules, resp.:
\begin{description}
 \item[]{\Expansion}
    \[
             (\cE\uplus\{ \CEqn[\simeq]{s}{t}{\varphi} \},\cH) \mathrel{\RIstep} (\cE\cup \Expd_\cR(\CEqn{s}{t}{\varphi},p), \cH\cup\{ \CRule{s}{t}{\varphi} \})
    \]
        where 
        \begin{itemize}
        	\item $p$ is a \emph{basic} position of $s$,%
        	\footnote{
        	A position of $p$ of term $s$ is \emph{basic} if $s|_p$ is of the form $f(s_1,\ldots,s_n)$ with $f$ a defined symbol and $s_1,\ldots,s_n$ constructor terms.} 
	        \item $\cR\cup\cH\cup\{ \CRule{s}{t}{\varphi} \}$ is terminating, and
       		\item $\Expd_\cR(\CEqn{s}{t}{\varphi},p)$ is the set of constrained equation $\CEqn{s'}{t'}{\varphi'}$ such that $\CEqn{s\gamma}{t\gamma}{\varphi\gamma \wedge \psi\gamma} \mathrel{\to_{1.p,\CRule{\ell}{r}{\psi}}} \CEqn{s'}{t'}{\varphi'}$ for some renamed variant $\CRule{\ell}{r}{\psi}$ of a rule in $\cR$ (i.e., $\FVar(\ell,r,\psi)\cap \FVar(s,t,\varphi)=\emptyset$) and a most general unifier $\gamma$ of $s|_p$ and $\ell$.
       	\end{itemize}
Note that $\approx$ is considered a binary function symbol in constrained rewriting.


 \item[]{\Simplification}
    \[
        (\cE\uplus\{ \CEqn[\simeq]{s}{t}{\varphi} \},\cH) \mathrel{\RIstep} (\cE\cup\{ \CEqn{u}{t}{\psi} \}, \cH)
    \]
        where $\CTerm{s}{\varphi} \mathrel{\to_{\cR\cup \cH}} \CTerm{u}{\psi}$.

 \item[]{\Deletion}
    \[
            (\cE\uplus\{ \CEqn[\approx]{s}{t}{\varphi} \},\cH) \mathrel{\RIstep} (\cE,\cH)
    \]
        where $s = t$ or $\varphi$ is not satisfiable.
        

\end{description}
In addition to the above, we use the following inference rules:
\begin{description}
 \item[]{\CaseSplitting}
   \[
             (\cE\uplus\{ \CEqn[\simeq]{s}{t}{\varphi} \},\cH) \mathrel{\RIstep} (\cE\cup \Expd_\cR(\CEqn{s}{t}{\varphi},p), \cH)
    \]
        where 
        	$p$ is a \emph{basic} position of $s$.
        Note that {\CaseSplitting} is a variant of {\Expansion} without adding  $\CRule{s}{t}{\varphi}$ to $\cH$.

	\item[]{\Generalization}
	\[
	(\cE\uplus\{ \CEqn{s}{t}{\varphi} \}, \cH) \mathrel{\RIstep} (\cE\cup\{ \CEqn{s}{t}{\psi} \}, \cH)
	\]
	where $\varphi \implies \psi$ is valid.
	Note that this is a simpler version of the original one in~\cite{FKN17tocl}.
\end{description}
A pair $(\cE,\cH)$ is called a \emph{process} of RI.
Starting with $(\cE,\emptyset)$, we apply the inference rules above to  processes of RI.
If we get $(\emptyset,\cH)$, then all the equations in $\cE$ are proved to be inductive theorems of $\cR$.

Next, we revisit the role of termination in the RI method.
When we apply {\Expansion} to $(\cE_i,\cH_i)$, we prove termination of $\cR \cup \cH_i \cup \{ \CRule{s}{t}{\varphi} \}$.
This is necessary to avoid both constructing an incorrect inference sequence and applying {\Simplification} infinitely many times.
However, from theoretical viewpoint, it suffices to prove termination of $\cR\cup\cH$ after constructing an inference sequence $(\cE,\emptyset) \mathrel{\RIstep} \cdots \mathrel{\RIstep} (\emptyset,\cH)$. 
In this paper, we drop termination of $\cR\cup\cH\cup\{\CRule{s}{t}{\varphi}\}$ from the side condition of {\Expansion}.
Due to this relaxation, 
a constructed inference sequence does not always ensure that $\cE$ is a set of inductive theorems of $\cR$.
For this reason, we introduce the notion of \emph{valid inference sequences}.
An inference sequence $(\cE,\emptyset) \mathrel{\RIstep} \cdots \mathrel{\RIstep} (\emptyset,\cH)$ is called \emph{valid} if $\cR\cup\cH$ is terminating. 
\begin{theorem}[\cite{FKN17tocl}]
Let $\cR$ be an LCTRS and $\cE$ a finite set of equations.
If we have a valid inference sequence $(\cE,\emptyset) \mathrel{\RIstep} \cdots \mathrel{\RIstep} (\emptyset,\cH)$, then every equation in $\cE$ is an inductive theorem of $\cR$.
\end{theorem}

\section{Transforming a Proof Tableau into an Inference Sequence of RI}
\label{sec: transform}

In this section, using the proof tableau $\Psumtab$, we first illustrate a transformation of a proof tableau into an inference sequence of RI, and then formalize the transformation.

\subsection{Overview}
\label{subsec:overview}

Let us recall the LCTRS $\cRsum$ in 
Example~\ref{ex:sum}
and the proof tableau $\Psumtab$ in Figure~\ref{fig:Ptab}.
To verify the post-condition after the execution of $\Psum$, we prepare the following rules with a new symbol $\Check: \sort{state} \funtype \BOOL$:  
\[
\cRsumcheck =
\left \{
\begin{array}{r@{\>}c@{\>}l@{\>}c@{}c@{}c}
\CRule{\Check(\symb{end}(x,i,z)) &}{& \symb{true} &}{& z = \frac{\symb{1}}{\symb{2}}x(x+\symb{1}) &}\\
\CRule{\Check(\symb{end}(x,i,z)) &}{& \symb{false} &}{& \neg (z = \frac{\symb{1}}{\symb{2}}x(x+\symb{1})) &}\\
\end{array}
\right\}
\]
We let $\cRsumverif=\cRsum \cup \cRsumcheck$.
To prove the Hoare triple $\triple{x \geq 0}{\Psum}{z = \frac{1}{2}x(x+\symb{1})}$ to hold, it suffices to consider initial states satisfying the pre-condition $x \geq 0$, and thus, we prove the following equation an inductive theorem of $\cRsumverif$:
\[
\mathrm{(A1)} ~~~~ \CEqn{\Check(\State_1(x,i,z))}{\symb{true}}{ x \geq 0 }
\]
It is clear that $\cRsumverif$ is quasi-reductive.

From now on, we 
transform the proof tableau $\Psumtab$ into an inference sequence of RI for $\cRsumverif$ in a \emph{top-down} fashion. 
The construction is independent of termination of $\cRsumverif$ with generated rules, and thus the construction itself does not ensure validity of the resulting inference sequence.

We start with the 
initial process
$(\{~\mathrm{(A1)}~\}, \emptyset)$.
Line \Lnum{A2} of $\Psumtab$ is an assertion $@\, {x \geq 0} \land {0 = 0}$ and the validity of ${x \geq 0} \implies {x \geq 0} \land {0 = 0}$ is guaranteed by the fact that $\Psumtab$ is a proof tableau.
Using the validity, we can generalize 
(A1) by applying {\Generalization} to the above process:
\[
\left( \left\{
~\mathrm{(A2)} ~~~~ \CEqn{\Check(\State_1(x,i,z))}{\symb{true}}{{x \geq \symb{0}} \Emph{{} \land {\symb{0} = \symb{0}}}}
\right\} ,  \emptyset \right)
\]

Let us recall the inference rule of assignment in Hoare logic (Figure~\ref{fig:Hoare}).
For an assignment $x_k := e$ on line \Lnum{$j$}, a rewrite rule $\Rule{\State_j(\vec{x})}{\State_{j+1}(x_1,\ldots,x_{k-1},e,x_{k+1},\ldots,x_n)}$ is generated, and thus, we have 
the derivation
$\CTerm{\State_j(\vec{x})}{\varphi\{x_k\mapsto e\}} \mathrel{\to_\cR} \CTerm{\State_{j+1}(\vec{x})}{\varphi}$
because 
$ 
\CTerm{\State_j(\vec{x})}{\varphi\{x_k\mapsto e\}} \mathrel{\to_{\mathtt{base}}} \CTerm{\State_{j+1}(x_1,\ldots,x_{k-1},e,x_{k+1},\ldots,x_n)}{\varphi\{x_k\mapsto e\}} \mathrel{\sim} 
\CTerm{\State_{j+1}(\vec{x})}{\varphi}
$. 
Line 1 of $\Psumtab$ is an assignment $i := 0$, 
and hence, $\CTerm{\State_1(x,i,z)}{{x \geq \symb{0}} \land {\symb{0} = \symb{0}}} \mathrel{\to_{\cRsumverif}} \CTerm{\State_2(x,i,z)}{{x \geq \symb{0}} \land {i = \symb{0}}}$.
Thus, we can simplify 
(A2) by applying {\Simplification} to the above process:
\[
\left( \left\{
~\mathrm{(A3)} ~~~~ \CEqn{\Check(\Emph{\State_2(x,i,z)})}{\symb{true}}{x \geq \symb{0} \land \Emph{i} = \symb{0}}
\right\} ,  \emptyset \right)
\]
Line \Lnum{A4} of $\Psumtab$ is $@\, x \geq 0 \land i = 0 \land 0 = 0$ 
and we can generalize 
(A3) by applying {\Generalization}: 
\[
\left( \left\{
~\mathrm{(A4)} ~~~~ \CEqn{\Check(\State_2(x,i,z))}{\symb{true}}{x \geq \symb{0} \land i = \symb{0} \Emph{{} \land {\symb{0} = \symb{0}}}}
\right\} ,  \emptyset \right)
\]
Line 2 of $\Psumtab$ is an assignment $z := 0$, 
and we can simplify 
(A4) by applying {\Simplification}: 
\[
\left( \left\{
~\mathrm{(A5)} ~~~~ \CEqn{\Check(\Emph{\State_3(x,i,z)})}{\symb{true}}{x \geq \symb{0} \land i = \symb{0} \land \Emph{z} = \symb{0}}
\right\} ,  \emptyset \right)
\]
Line \Lnum{A6} of $\Psumtab$ is $@\, z = \frac{1}{2}i(i+1) \land x \geq i$
and we can generalize 
(A5) by applying {\Generalization}: 
\[
\left( \left\{
~\mathrm{(A6)} ~~~~ \CEqn{\Check(\State_3(x,i,z))}{\symb{true}}{\Emph{z = \textstyle \frac{\symb{1}}{\symb{2}}i(i+\symb{1}) \land x \geq i}}
\right\} ,  \emptyset \right)
\]

Line 3 of $\Psumtab$ is a ``while'' statement. 
At this point, we have two branches:
the one entering the loop (i.e., executing the body of the loop) and 
the other exiting the loop.
For the case analysis, we apply {\Expansion} to 
(A6), getting the following two equations and one oriented equation:
\[
\left( 
\left\{
\begin{array}{r@{~~~~}r@{\>}c@{\>}l@{\>}c@{}c@{}c}
\Emph{\mathrm{(A7)}} & 
\CEqn{\Check(\State_4(x,i,z)) &}{&\symb{true} &}{&  z = \frac{\symb{1}}{\symb{2}}i(i+\symb{1}) \land x \geq i \land {x > i} &}\\
\Emph{\mathrm{(A11)}} & 
\CEqn{\Check(\symb{end}(x,i,z)) &}{&\symb{true} &}{& z = \frac{\symb{1}}{\symb{2}}i(i+\symb{1}) \land x \geq i \land \lnot (x > i)&}
\end{array}
\right\},  
\left\{
~
\Emph{\mathrm{(A6)}}
~
\right\} 
 \right)
\]
where (A6) is oriented from left to right.
The first equation represents the case where the loop body is executed, and the second one represents the case where we exit from the loop.

Line \Lnum{A8} of $\Psumtab$ is an assertion and we can generalize 
(A7) by applying {\Generalization}: 
\[
\left( 
\left\{
\begin{array}{r@{~~~~}r@{\>}c@{\>}l@{\>}c@{}c@{}c@{,~~}c}
\mathrm{(A8)} & \CEqn{\Check(\State_4(x,i,z)) &}{&\symb{true} &}{& \Emph{z + i + \symb{1} = \frac{\symb{1}}{\symb{2}}(i+\symb{1})(i+\symb{2}) \land {x \geq i+\symb{1}}}&}
& \mathrm{(A11)} 
\end{array}
\right\},  
\left\{
~
\mathrm{(A6)}
~
\right\} 
 \right)
\]
Line 4 of $\Psumtab$ is an assignment $z := z + i + 1$
and we can simplify 
(A8) by applying {\Simplification}: 
\[
\left( 
\left\{
\begin{array}{r@{~~~~}r@{\>}c@{\>}l@{\>}c@{}c@{}c@{,~~}c}
\mathrm{(A9)} & \CEqn{\Check(\Emph{\State_5(x,i,z)}) &}{&\symb{true} &}{& \Emph{z} = \frac{\symb{1}}{\symb{2}}(i+\symb{1})(i+\symb{2}) \land x \geq i+\symb{1}&}
&
\mathrm{(A11)} 
\end{array}
\right\} ,  
\left\{
~
\mathrm{(A6)}
~
\right\} 
 \right)
\]
Line 5 of $\Psumtab$ is an assignment $i := i + 1$
and we can simplify 
(A9) by applying {\Simplification}: 
\[
\left( 
\left\{
\begin{array}{r@{~~~~}r@{\>}c@{\>}l@{\>}c@{}c@{}c@{,~~}c}
\mathrm{(A10)} & \CEqn{\Check(\Emph{\State_6(x,i,z)}) &}{&\symb{true} &}{& z = \frac{\symb{1}}{\symb{2}}\Emph{i}(\Emph{i+\symb{1}}) \land x \geq \Emph{i} &}
&
\mathrm{(A11)} 
\end{array}
\right\} ,  
\left\{
~
\mathrm{(A6)}
~
\right\} 
 \right)
\]
Line 6 of $\Psumtab$ is the end of the loop and we can apply the rule $\Rule{\State_6(x,i,z)}{\State_3(x,i,z)}$ that makes the left-hand side of (A10) go back to the beginning of the loop.
Thus, we can simplify 
(A10): 
\[
\left( 
\left\{
\begin{array}{r@{~~~~}r@{\>}c@{\>}l@{\>}c@{}c@{}c@{,~~}c}
\mathrm{(B1)} & \CEqn{\Check(\Emph{\State_3(x,i,z)}) &}{&\symb{true} &}{& z = \frac{\symb{1}}{\symb{2}}i(i+\symb{1}) \land x \geq i &}
&
\mathrm{(A11)} 
\end{array}
\right\} ,  
\left\{
~
\mathrm{(A6)}
~
\right\} 
 \right)
\]
The equation (B1) means that we reach the beginning of the loop after the one execution of the body.
Moreover, 
(B1) is the same as (A6) due to the loop invariant, and hence the induction hypothesis (A6) is applicable to (B1).
Thus, we can simplify 
(B1) by applying {\Simplification} to the above process with rule (A6) $\CRule{\Check(\State_3(x,i,z))}{\symb{true}}{z = \frac{\symb{1}}{\symb{2}}i(i+\symb{1}) \land x \geq i}$:
\[
\left( 
\left\{
\begin{array}{r@{~~~~}r@{\>}c@{\>}l@{\>}c@{}c@{}c@{,~~}c}
\mathrm{(B2)} & \CEqn{\Emph{\symb{true}} &}{&\symb{true} &}{& z = \frac{\symb{1}}{\symb{2}}i(i+\symb{1}) \land x \geq i &}
&
\mathrm{(A11)} 
\end{array}
\right\} ,  
\left\{
~
\mathrm{(A6)}
~
\right\} 
\right)
\]
The both sides of 
(B2) are equivalent
and we can delete 
(B2) by applying {\Deletion}: 
\[
\left( 
\left\{
\begin{array}{r@{~~~~}r@{\>}c@{\>}l@{\>}c@{}c@{}c}
\mathrm{(A11)} & \CEqn{\Check(\symb{end}(x,i,z)) &}{&\symb{true} &}{& z = \frac{\symb{1}}{\symb{2}}i(i+\symb{1}) \land x \geq i \land \lnot (x > i)&}
\end{array}
\right\} ,  
\left\{
~
\mathrm{(A6)}
~
\right\} 
\right)
\]

The remaining equation (A11) represents the state after exiting the loop.
The last line of $\Psumtab$ is an assertion corresponding to the post-condition.
Due to the validity of $z = \frac{1}{2}i(i + 1) \land x \geq i \land \neg (x > i) \implies z = \frac{1}{2}x(x + 1)$, we can generalize 
(A11) by applying {\Generalization}: 
\[
\left( 
\left\{
\begin{array}{r@{~~~~}r@{\>}c@{\>}l@{\>}c@{}c@{}c}
\mathrm{(B3)} & 
\CEqn{\Check(\symb{end}(x,i,z)) &}{&\symb{true} &}{& \Emph{z = \frac{\symb{1}}{\symb{2}}x(x+\symb{1})} &}
\end{array}
\right\} ,  
\left\{
~
\mathrm{(A6)}
~
\right\} 
\right)
\]
The constraints of 
(B3) and the post-condition of $\Psumtab$ are equivalent and we can apply the first rule of $\cRsumcheck$ to the left-hand side of (B3) in order to verify the post-condition.
Thus, we can simplify 
(B3) by applying {\Simplification} with rule $\CRule{\Check(\symb{end}(x,i,z))}{\symb{true}}{z = \frac{\symb{1}}{\symb{2}}x(x+\symb{1})}$ in $\cRsumcheck$ to the above process:
\[
\left( 
\left\{
\begin{array}{r@{~~~~}r@{\>}c@{\>}l@{\>}c@{}c@{}c}
\mathrm{(B4)} & 
\CEqn{\Emph{\symb{true}} &}{&\symb{true} &}{& z = \frac{\symb{1}}{\symb{2}}x(x+\symb{1}) &}
\end{array}
\right\} ,  
\left\{
~
\mathrm{(A6)}
~
\right\} 
\right)
\]
The both sides of 
(B4) are equivalent 
and we can delete 
(B4) by applying {\Deletion}: 
\[
\left( 
\emptyset
,  
\left\{
~
\mathrm{(A6)}
~
\right\} 
\right)
\]
In the above illustration, we did not show the case of ``if'' statements.
However, the missing case is a simpler one of ``while'' statements, where we use {\CaseSplitting} instead of {\Expansion}.

Finally, we show that $\cRsumverif\cup \{~\mathrm{(A6)}~\}$ is terminating.
Since any term with sort $\sort{state}$ or $\BOOL$ does not appear in $\cRsum$ as a proper subterm, 
$\cRsumcheck\cup \{~\mathrm{(A6)}~\}$ does not introduce non-termination into 
$\cRsum$.
As described before, $\cRsum$ is terminating and hence $\cRsumverif\cup \{~\mathrm{(A6)}~\}$ is so.


\subsection{Formalization}

In this section, we formalize the idea illustrated in the previous section.
In the following, we consider 
\begin{itemize}
	\item a \emph{while} program $P$ such that $\FVar(P) = \{ x_1,\ldots,x_n \}$,
	\item a proof tableau $T_P$ for a Hoare triple $\triple{\varphi_P}{P}{\psi_P}$,%
	\footnote{
	Note that $P$ is the same as the \emph{while} program obtained from $T_P$ by removing assertions.}
			and
	\item the LCTRS $\cR_P$ obtained from $P$ by the conversion in Section~\ref{subsec:while_to_lctrs}.
\end{itemize}
We denote the sequence $x_1,\ldots,x_n$ by $\vec{x}$.
Unlike previous sections, we specify line numbers for $T_P$, and reuse them in converting $P$ to $\cR_P$.
For this reason, the function symbol to represent initial states is not $\State_1$ but $\State_{i_0}$ for some $i_0 > 1$.
Notice that the pre-condition $\varphi_P$ is on line $1$ of $T_P$ as an assertion.
For readability, we use $\Start$ as a meta symbol that stands for $\State_{i_0}$.


To check whether the final state of the execution of $P$ satisfies the post-condition $\psi_P$, we prepare the following rules:
\[
\cRcheck =
\left \{
\begin{array}{r@{\>}c@{\>}l@{\>}c@{}c@{}c}
\CRule{\Check(\symb{end}(\vec{x})) &}{& \symb{true} &}{& \psi_P &}\\
\CRule{\Check(\symb{end}(\vec{x})) &}{& \symb{false} &}{& \neg \psi_P &}\\
\end{array}
\right\}
\]
where $\Check: \sort{state} \funtype \BOOL$.
Then, to verify the Hoare triple $\triple{\varphi_P}{P}{\psi_P}$, we prepare the following constrained equation: 
\[
\CEqn{\Check(\Start(\vec{x}))}{\symb{true}}{ \varphi_P }
\]
In the following, we denote the above equation by $\EP$.

By definition, $\cR_P\cup\cRcheck$ has the following properties.
\begin{lemma}
\label{lem:Rcheck-properties}
All of the following hold:
\begin{enumerate}
\renewcommand{\labelenumi}{(\alph{enumi})}
\renewcommand{\theenumi}{(\alph{enumi})}
	\item\label{lem:item:Rcheck-orthogonal}
		$\cR_P \cup \cRcheck$ is orthogonal.
	\item\label{lem:item:termination_of_Rcheck}
		If $\cR_P$ is terminating, then $\cR_P\cup\cRcheck$ is so.
\end{enumerate}
\end{lemma}
\begin{proof}
We first prove \ref{lem:item:Rcheck-orthogonal}.
As described in Section~\ref{subsec:while_to_lctrs}, $\cR_P$ is orthogonal.
By definition, $\cRcheck$ is orthogonal.
$\cRcheck$ has no defined symbol of $\cR_P$ and thus, $\cRcheck$ does not generate any overlap with $\cR_P$.
Therefore, $\cR_P\cup\cRcheck$ is orthogonal. 

Next, we prove \ref{lem:item:termination_of_Rcheck}.
Assume that $\cR_P$ is terminating but $\cR_P\cup\cRcheck$ is not.
Then, there exists an infinite reduction sequence of $\cR_P\cup\cRcheck$.
Due to the sort of $\Check$, the infinite reduction sequence starts with a term of the form $\State_i(\vec{t})$, and any rule of $\cRcheck$ is not used in the reduction sequence. 
This means that the infinite reduction sequence is caused by $\cR_P$.
This contradicts the assumption.
\qed	
\end{proof}

The equation $\EP$ has the following property.
\begin{theorem}
\label{thm:e_P}
If $\EP$ is an inductive theorem of\/ $\cR_P\cup\cRcheck$, then $\triple{\varphi_P}{P}{\psi_P}$ holds.	
\end{theorem}
\begin{proof}
Let $\theta$ be an assignment for $\FVar(P)$.
Assume that $\varphi_P\theta$ holds and the execution of $P$ starting with $\theta$ halts with an assignment $\theta'$, i.e., $\theta \mathrel{\Eval{P}} \theta'$.
Then, it follows from Theorem~\ref{thm:correctness_of_conversion} that $\Start(\vec{x})\theta \mathrel{\to^*_{\cR_P}} \symb{end}(\vec{x})\theta'$, and hence $\Check(\Start(\vec{x}))\theta \mathrel{\to^*_{\cR_P}} \Check(\symb{end}(\vec{x}))\theta'$.
Since $\varphi_P\theta$ holds and $\EP$ is an inductive theorem of $\cR_P\cup\cRcheck$, we have that $\Check(\Start(\vec{x}))\theta \mathrel{\leftrightarrow^*_{\cR_P\cup\cRcheck}} \symb{true}$.
Since $\cR_P\cup\cRcheck$ is orthogonal (i.e., confluent) by Lemma~\ref{lem:Rcheck-properties}~\ref{lem:item:Rcheck-orthogonal}, we have that $\Check(\Start(\vec{x}))\theta \mathrel{\to^*_{\cR_P\cup\cRcheck}} \symb{true}$ and hence
$\Check(\symb{end}(\vec{x}))\theta'$ has to reduce to $\symb{true}$.
This means that $\psi_P\theta'$ holds.
Therefore, $\triple{\varphi_P}{P}{\psi_P}$ holds.
\qed	
\end{proof}
Theorem~\ref{thm:e_P} enables us to prove $\triple{\varphi_P}{P}{\psi_P}$ to hold by showing that $\EP$ is an inductive theorem of $\cR_P\cup\cRcheck$.
Note that the converse of Theorem~\ref{thm:e_P} holds if $P$ is terminating. 

Next, we formalize the transformation shown in Section~\ref{subsec:overview}.
We first prepare a function $\TransOneStep$ that takes a suffix $T$ of proof tableau $T_P$ and finite sets $\cE$ and $\cH$ of equations and rewrite rules, resp.,
and returns a suffix $T'$ of $T$, and finite sets $\cE'$ and $\cH'$ of equations and rewrite rules, resp.:
$\TransOneStep(T,\cE,\cH)=(T',\cE',\cH')$.
For readability, we use visualized notations for suffixes of proof tableaux, e.g., 
\begin{center}
\smallskip
{\footnotesize
\begin{tabular}{|@{\,}c@{\,}@{~~~}@{\,}l@{\,}|}
\cline{1-2}
$i$ & $@\varphi$;~~~~~~~~~\\
$i+1$& $@\psi$;\\
\vdots & ~~~\vdots\\
\cline{1-2}
\end{tabular}
}
\end{center}
\smallskip
for $@\varphi;\,@\psi;\,\ldots$ such that the first element $@\varphi$ is located on line $i$.
We assume that any equation in $\cE$ is of the form $\CEqn{\Check(\State_j(\vec{t}))}{\symb{true}}{\varphi}$ or $\CEqn{\Check(\symb{end}(\vec{t}))}{\symb{true}}{\varphi}$,
and then we define $\TransOneStep$ so as to make $\cE'$ a set of such equations.
Following the definition of proof tableaux,
the function $\TransOneStep$ is defined as follows:
\begin{itemize}
\item (two continuous assertions)
\[\begin{array}{@{}l@{}}
\TransOneStep(
\,
\mbox{\footnotesize
\begin{tabular}{|@{\,}c@{\,}@{~~~}@{\,}l@{\,}|}
\cline{1-2}
$i$ & $@\varphi$;~~~~~~~~~\\
$i+1$& $@\psi$;\\
\vdots & ~~~\vdots\\
\cline{1-2}
\end{tabular}
}
\, , \,
\{~\CEqn{\Check(\State_j(\vec{x}))}{\symb{true}}{\varphi}~\}\uplus \cE
, \,
\cH
) \\[18pt]
\hspace{10ex}
=
(
\,
\mbox{\footnotesize
\begin{tabular}{|@{\,}c@{\,}@{~~~}@{\,}l@{\,}|}
\cline{1-2}
$i+1$& $@\psi$;~~~~~~~~~\\
\vdots & ~~~\vdots\\
\cline{1-2}
\end{tabular}
}
\, ,
\{~\CEqn{\Check(\State_j(\vec{x}))}{\symb{true}}{\textcolor{blue}{\psi}}~\}\cup \cE
, 
\cH
)
\end{array}
\]
Note that $i > j$.
This case corresponds to the application of {\Generalization} to $(\{~\CEqn{\Check(\State_j(\vec{x}))$ $}{\symb{true}}{\psi}~\}\uplus \cE, \cH)$.

\item (assignments)
\[\begin{array}{@{}l@{}}
\TransOneStep(
\,
\mbox{\footnotesize
\begin{tabular}{|@{\,}c@{\,}@{~~~}@{\,}l@{\,}|}
\cline{1-2}
$i$ & $@\varphi$;~~~~~~~~~\\
$i+1$& $x_k := e$;\\
$i+2$ & $@\psi$;\\
\vdots & ~~~\vdots\\
\cline{1-2}
\end{tabular}
}
\,, \,
\{~\CEqn{\Check(\State_{i+1}(\vec{x}))}{\symb{true}}{\varphi}~\}\uplus \cE
, \,
\cH
) \\[25pt]
\hspace{10ex}
= 
(
\,
\mbox{\footnotesize
\begin{tabular}{|@{\,}c@{\,}@{~~~}@{\,}l@{\,}|}
\cline{1-2}
$i+2$ & $@\psi$;~~~~~~~~~\\
\vdots & ~~~\vdots\\\cline{1-2}
\end{tabular}
}
\, ,
\{~\CEqn{\Check(\State_{\textcolor{blue}{j}}(\vec{x}))}{\symb{true}}{\textcolor{blue}{\psi}}~\}\cup \cE
,
\cH
)
\end{array}
\]
where $\Rule{\State_{i+1}(\ldots,x_k,\ldots)}{\State_j(\ldots,e,\ldots)} \in \cR_P$.
Note that $\varphi = \psi\{x_k\mapsto e\}$, $j> i+1$, and $\CTerm{\State_{i+1}(\vec{x})}{\varphi} \mathrel{\to_{\cR_P}} \CTerm{\State_j(\vec{x})}{\psi}$.%
\footnote{
This is because 
$ 
\CTerm{\State_{i+1}(\vec{x})}{\varphi}
=
\CTerm{\State_{i+1}(\vec{x})}{\psi\{x_k\mapsto e\}} 
\mathrel{\to_{\mathtt{base}}} \CTerm{\State_j(x_1,\ldots,x_{k-1},e,x_{k+1},\ldots,x_n)}{\psi\{x_k\mapsto e\}} \mathrel{\sim} 
\CTerm{\State_j(\vec{x})}{\psi}
$. 
}
This case corresponds to the application of {\Simplification} to $(\{~\CEqn{\Check(\State_{i+1}(\vec{x}))}{\symb{true}}{\varphi}~\}\uplus \cE, \cH)$.

\item (the beginning of ``while'' statements)
\[\begin{array}{@{}l@{}}
\TransOneStep(
\,
\mbox{\footnotesize
\begin{tabular}{|@{\,}c@{\,}@{~~~}@{\,}l@{\,}|}
\cline{1-2}
$i$ & $@\zeta$;~~~~~~~~~\\
$i+1$& $\while @\,\xi\,(\varphi)\{$\\
$i+2$ & ~~~~$@\zeta \land \varphi$;\\
\vdots & ~~~~~~\vdots\\
\cline{1-2}
\end{tabular}
}
\, , \,
\{~\CEqn{\Check(\State_{i+1}(\vec{x}))}{\symb{true}}{\zeta}~\}\uplus \cE
, \,
\cH
) \\[25pt]
~~
=
(
\,
\mbox{\footnotesize
\begin{tabular}{|@{\,}c@{\,}@{~~~}@{\,}l@{\,}|}
\cline{1-2}
$i+2$ & ~~~~$@\zeta\land \varphi$;~~~~~~~~~\\
\vdots & ~~~~~~\vdots\\\cline{1-2}
\end{tabular}
}
\,, 
\left\{
\begin{array}{r@{\>}c@{\>}l@{\>}c@{}c@{}c}
\CEqn{\Check(\State_{\textcolor{blue}{j}}(\vec{x})) &}{& \symb{true} &}{& \zeta\land \textcolor{blue}{\varphi} &}\\
\CEqn{\Check(\State_{\textcolor{blue}{k}}(\vec{x})) &}{& \symb{true} &}{& \zeta\land \textcolor{blue}{\neg \varphi} &}\\
\end{array}
\right\}
\cup 
\cE
, \\\hspace{50ex}
\{~\textcolor{blue}{\CRule{\Check(\State_{i+1}(\vec{x}))}{\symb{true}}{\zeta}}~\}\cup\cH
)
\end{array}
\]
where $\CRule{\State_{i+1}(\vec{x})}{\State_j(\vec{x})}{\varphi},
\CRule{\State_{i+1}(\vec{x})}{\State_k(\vec{x})}{\neg \varphi} \in \cR_P$.
Note that $i+1 < j < k$.
This case corresponds to the application of {\Expansion} to $(\{~\CEqn{\Check(\State_{i+1}(\vec{x}))}{\symb{true}}{\zeta}~\}\uplus \cE, \cH)$.

\item (the end of ``while'' statements)
\[
\TransOneStep(
~
\mbox{\footnotesize
\begin{tabular}{|@{\,}c@{\,}@{~~~}@{\,}l@{\,}|}
\cline{1-2}
$i$ & ~~~~ $@\zeta$;~~~~~~~~~\\
$i+1$& $\}$\\
$i+2$ & $@\zeta \land \neg \varphi$;\\
\vdots & ~~~\vdots\\
\cline{1-2}
\end{tabular}
}
\, , \,
\{~\CEqn{\Check(\State_{i+1}(\vec{x}))}{\symb{true}}{\zeta}~\}
\uplus \cE
, \,
\cH
) 
=
(
\,
\mbox{\footnotesize
\begin{tabular}{|@{\,}c@{\,}@{~~~}@{\,}l@{\,}|}
\cline{1-2}
$i+2$ & $@\zeta\land \neg \varphi$;~~~~\\
\vdots & ~~~\vdots\\\cline{1-2}
\end{tabular}
}
\,, 
\cE
, \cH
)
\]
where $\Rule{\State_{i+1}(\vec{x})}{\State_j(\vec{x})} \in \cR_P$, $\CRule{\Check(\State_j(\vec{x}))}{\symb{true}}{\zeta} \in \cH$,  and $j < i+1$.
This case corresponds to the application of {\Simplification} with rule $\Rule{\State_{i+1}(\vec{c})}{\State_j(\vec{x})} \in \cR_P$,
{\Simplification} with rule $\CRule{\Check(\State_j(\vec{x}))}{\symb{true}}{\zeta} \in \cH$,
 and
{\Deletion}:
\[
\begin{array}{@{}l@{~}l@{}}
(\{\CEqn{\Check(\State_{i+1}(\vec{x}))}{\symb{true}}{\zeta}\}\cup \cE, \cH) 
&
{} \mathrel{\RIstep}
(\{\CEqn{\Check(\State_j(\vec{x}))}{\symb{true}}{\zeta}\}\cup \cE, \cH) \\ 
&
{} \mathrel{\RIstep}
(\{\CEqn{\symb{true}}{\symb{true}}{\zeta}\}\cup \cE, \cH) \\
&
{} \mathrel{\RIstep}
(\cE, \cH)
\\
\end{array}
\]

\item (the beginning of ``if'' statements)
\[\begin{array}{@{}l@{}}
\TransOneStep(
\,
\mbox{\footnotesize
\begin{tabular}{|@{\,}c@{\,}@{~~~}@{\,}l@{\,}|}
\cline{1-2}
$i$ & $@\varphi$;~~~~~~~~~\\
$i+1$& $\ifff(\psi)\{$\\
$i+2$ & ~~~~$@\varphi \land \psi$;\\
\vdots & ~~~~~~\vdots\\
\cline{1-2}
\end{tabular}
}
\, , \,
\{~\CEqn{\Check(\State_{i+1}(\vec{x}))}{\symb{true}}{\varphi}~\}\uplus \cE
, \,
\cH
) \\[25pt]
\hspace{10ex}
= 
(
\,
\mbox{\footnotesize
\begin{tabular}{|@{\,}c@{\,}@{~~~}@{\,}l@{\,}|}
\cline{1-2}
$i+2$ & ~~~~$@\varphi\land \psi$;~~~~\\
\vdots & ~~~~~~\vdots\\\cline{1-2}
\end{tabular}
}
\, ,
\left\{
\begin{array}{r@{\>}c@{\>}l@{\>}c@{}c@{}c}
\CEqn{\Check(\State_{\textcolor{blue}{j}}(\vec{x})) &}{& \symb{true} &}{& \varphi\land \textcolor{blue}{\psi} &}\\
\CEqn{\Check(\State_{\textcolor{blue}{k}}(\vec{x})) &}{& \symb{true} &}{& \varphi\land \textcolor{blue}{\neg \psi} &}\\
\end{array}
\right\}
\cup \cE
, 
\cH
)
\end{array}
\]
where $\CRule{\State_{i+1}(\vec{x})}{\State_j(\vec{x})}{\psi},
\CRule{\State_{i+1}(\vec{x})}{\State_k(\vec{x})}{\neg \psi} \in \cR_P$.
Note that $i+1 < j < k$.
This case corresponds to the application of {\CaseSplitting} to $(\{~\CEqn{\Check(\State_{i+1}(\vec{x}))}{\symb{true}}{\varphi}~\}\uplus \cE, \cH)$.

\item (the beginning of ``else'' statements)
\[\begin{array}{@{}l@{}}
\TransOneStep(
\,
\mbox{\footnotesize
\begin{tabular}{|@{\,}c@{\,}@{~~~}@{\,}l@{\,}|}
\cline{1-2}
$i$ & ~~~~$@\xi$;~~~~~~~~~\\
$i+1$& $\} \elseee \{$\\
$i+2$ & ~~~~$@\varphi \land \neg \psi$;\\
\vdots & ~~~~~~\vdots\\
\cline{1-2}
\end{tabular}
}
\, , \,
\left\{
\begin{array}{r@{\>}c@{\>}l@{\>}c@{}c@{}c}
\CEqn{\Check(\State_{i+1}(\vec{x})) &}{& \symb{true} &}{& \xi &}\\
\CEqn{\Check(\State_{j}(\vec{x})) &}{& \symb{true} &}{& \varphi\land \neg \psi &}\\
\end{array}
\right\}
\uplus \cE
, \,
\cH
) \\[25pt]
\hspace{10ex}
=
(
\,
\mbox{\footnotesize
\begin{tabular}{|@{\,}c@{\,}@{~~~}@{\,}l@{\,}|}
\cline{1-2}
$i+2$ & ~~~~$@\varphi\land \neg \psi$;~~~~\\
\vdots & ~~~~~~\vdots\\\cline{1-2}
\end{tabular}
}
\, ,
 \left\{
\begin{array}{r@{\>}c@{\>}l@{\>}c@{}c@{}c}
\CEqn{\Check(\State_{\textcolor{blue}{k}}(\vec{x})) &}{& \symb{true} &}{& \xi &}\\
\CEqn{\Check(\State_{j}(\vec{x})) &}{& \symb{true} &}{& \varphi\land \neg \psi &}\\
\end{array}
\right\}
\cup \cE
,
\cH
)
\end{array}
\]
where $\Rule{\State_{i+1}(\vec{x})}{\State_k(\vec{x})} \in \cR_P$.
Note that $i+2 < j < k$.
This case corresponds to the application of {\Simplification}
 to 
 $(
\{~
\CEqn{\Check(\State_{i+1}(\vec{x}))}{\symb{true}}{\xi},
~ 
\CEqn{\Check(\State_{j}(\vec{x}))}{\symb{true}}{\varphi\land \neg \psi}
~\}
\uplus \cE
,
\cH)$.

\item (the end of ``if'' statements)
\[\begin{array}{@{}l@{}}
\TransOneStep(
\,
\mbox{\footnotesize
\begin{tabular}{|@{\,}c@{\,}@{~~~}@{\,}l@{\,}|}
\cline{1-2}
$i$ & ~~~~$@\xi$;~~~~~~~~~\\
$i+1$& $\}$\\
$i+2$ & $@\xi$;\\
\vdots & ~~~\vdots\\
\cline{1-2}
\end{tabular}
}
\, , \,
\left\{
\begin{array}{r@{\>}c@{\>}l@{\>}c@{}c@{}c}
\CEqn{\Check(\State_k(\vec{x})) &}{& \symb{true} &}{& \xi &}\\
\CEqn{\Check(\State_{i+1}(\vec{x})) &}{& \symb{true} &}{& \xi &}\\
\end{array}
\right\}\uplus \cE
 , \,
\cH
) \\[25pt]
\hspace{10ex}
=
(
\,
\mbox{\footnotesize
\begin{tabular}{|@{\,}c@{\,}@{~~~}@{\,}l@{\,}|}
\cline{1-2}
$i+2$ & $@\xi$;~~~~~~~~~\\
\vdots & ~~~\vdots\\\cline{1-2}
\end{tabular}
}
\, , 
 \{~\CEqn{\Check(\State_{k}(\vec{x}))}{\symb{true}}{\xi}~\}
\cup \cE
, 
 \cH
)
\end{array}
\]
where $\Rule{\State_{i+1}(\vec{x})}{\State_k(\vec{x})} \in \cR_P$ and $i+1< k$.
This case corresponds to the application of {\Simplification}
 to 
 $(
\{~
\CEqn{\Check(\State_k(\vec{x}))}{\symb{true}}{\xi},
~
\CEqn{\Check(\State_{i+1}(\vec{x}))}{\symb{true}}{\xi}
~\}
\uplus \cE
 ,
\cH
)$.

\item (the end of tableaux)
\[
\TransOneStep(
\,
\mbox{\footnotesize
\begin{tabular}{|@{~}c@{~}@{~~~}@{~}l@{~}|}
\cline{1-2}
$i$ & $@\varphi$;~~~~~~~~~\\
\cline{1-2}
\end{tabular}
}
\, , \,
 \{~\CEqn{\Check(\symb{end}(\vec{x}))}{\symb{true}}{\varphi}~\}
 \uplus \cE
  , \,
 \cH
)
=
(
\,
\begin{tabular}{|@{\,}c@{\,}|@{\,}l@{\,}|}
\cline{1-2}
\end{tabular}
 \epsilon 
 , 
 \cE
,
\cH
)
\]
Note that the last element of $T_P$ is $@\psi_P$ and thus, $\varphi=\psi_P$.
Note also that $\CRule{\Check(\symb{end}(\vec{x}))}{\symb{true}}{\psi_P} \in \cRcheck$.
This case corresponds to the application of {\Simplification} and {\Deletion}:
\[
(\{\CEqn{\Check(\symb{end}(\vec{x}))}{\symb{true}}{\varphi}\}\cup \cE, \cH)
\mathrel{\RIstep}
(\{\CEqn{\symb{true}}{\symb{true}}{\varphi}\}\cup \cE, \cH) 
\mathrel{\RIstep}
(\cE, \cH)
\]

\end{itemize}
By by the definition of proof tableaux and $\TransOneStep$,
$\TransOneStep$ satisfies the following properties.
\begin{lemma}\label{lem:TransOneStep}
If $\TransOneStep(T,\cE,\cH) = (T',\cE',\cH')$, then 
(a)	$(\cE,\cH) \mathrel{\RIstep^*} (\cE',\cH')$, and
(b) if $T' \neq \epsilon$, then $\TransOneStep$ is applicable to $(T',\cE',\cH')$.
\end{lemma}

Next, we define a function $\Trans$ that applies $\TransOneStep$ to $(T_P,\{\EP\},\emptyset)$ as much as possible, returning a list of RI processes:
\begin{itemize}
	\item $\Trans(\epsilon, \cE,\cH) = (\cE,\cH)$,
			and
	\item $\Trans(T,\cE,\cH) = (\cE,\cH), \Trans(T',\cE',\cH')$ where $T\ne\epsilon$ and $\TransOneStep(T,\cE,\cH)=(T',\cE',\cH')$.%
\footnote{
The result of $\Trans(T,\cE,\cH)$ is a sequence ``$(\cE,\cH), \Trans(T',\cE',\cH')$'' that has $(\cE,\cH)$ as its head element.
}
\end{itemize}
By definition and Lemma~\ref{lem:TransOneStep}, $\Trans$ satisfies the following properties.
\begin{lemma}
\label{lem:Trans-termination}
$\Trans(T_P,\{\EP\},\emptyset)$ returns a finite sequence of RI processes.
\end{lemma}
\begin{proof}
The first argument of $\Trans$ is a proof tableau and the length is decreasing when $\Trans$ is recursively called.
It follows from Lemma~\ref{lem:TransOneStep}~(b) that $\Trans$ calls $\TransOneStep$ until the first argument (suffixes of $T_P$) becomes $\epsilon$.
Therefore, $\Trans$ halts, returning a finite sequence of RI processes.
\qed
\end{proof}
\begin{lemma}
\label{lem:Trans}
Let the result of $\Trans(T_P,\{\EP\},\emptyset)$ be a sequence $(\cE_1,\cH_1), (\cE_2,\cH_2), \ldots, (\cE_n,\cH_n)$. 
Then,
$(\{\EP\},\emptyset) = (\cE_1,\cH_1) \mathrel{\RIstep^*} (\cE_2,\cH_2) \mathrel{\RIstep^*} \cdots \mathrel{\RIstep^*} (\cE_n,\cH_n) = (\emptyset,\cH_n) $.
\end{lemma}
\begin{proof}
By definition, it is clear that the head of the resulting sequence is $(\{\EP\},\emptyset)$.
The last call of $\Trans$ takes $\epsilon$ as the first argument, and thus the last call of $\TransOneStep$ returns $(\epsilon,\cE_n,\cH_n)$.
In the case of the beginning of ``while'' or ``if'' statements, $\TransOneStep$ adds an equation to the second argument, and in the case of the end of ``while'' or ``if'' statements, $\TransOneStep$ removes an equation from the second argument.
This means that in the case of the end of $T_P$, the number of remaining equations is one, i.e., $|\cE_{n-1}|=1$.
It follows from the last application $\TransOneStep(\ldots,\cE_{n-1},\cH_{n-1})=(\epsilon,\cE_n,\cH_n)$ that $\cE_{n-1}=\{\CEqn{\Check(\symb{end}(\vec{x}))}{\symb{true}}{\psi_P}\}$ and $\cE_n=\emptyset$.
It follows from Lemma~\ref{lem:TransOneStep}~(a) that $(\cE_i,\cH_i) \mathrel{\RIstep^*} (\cE_{i+1},\cH_{i+1})$ for all $1 \leq i < n$.
Therefore, this lemma holds.
\qed
\end{proof}


Finally, we show that termination of $\cR_P$ implies both termination of $\cR_P\cup\cRcheck\cup\cH$ and total correctness of $P$ w.r.t.\ $\varphi_P$ and $\psi_P$.
Let $\Trans(T_P,\{\EP\},\emptyset)=(\{\EP\},\emptyset),\ldots, (\emptyset,\cH)$.
We have already shown that termination of $\cR_P$ implies termination of $\cR_P\cup\cRcheck$ (Lemma~\ref{lem:Rcheck-properties}~\ref{lem:item:termination_of_Rcheck}).
Thus, we show that termination of $\cR_P$ implies termination of $\cR_P\cup\cRcheck\cup\cH$.
Since the right-hand sides of oriented equations in $\cH$ are always $\symb{true}$, 
$\cH$ is always terminating and does not introduce non-termination into $\cR_P \cup \cRcheck$.
This means that if $\cR_P\cup\cRcheck$ is terminating, then so is $\cR_P \cup \cRcheck \cup \cH$.
\begin{theorem}
\label{thm:termination}
If $\cR_P$ is terminating, then $\cR_P\cup\cRcheck\cup\cH$ is so.	
\end{theorem}

As a consequence of Lemma~\ref{lem:Trans} and Theorem~\ref{thm:termination}, we have the following result.
\begin{theorem}
\label{thm:main}
If $\cR_P$ is terminating, then	
$(\{\EP\},\emptyset) \mathrel{\RIstep^*} \cdots \mathrel{\RIstep^*} (\emptyset,\cH)$ is valid,
		and thus,
$\ttriple{\varphi_P}{P}{\psi_P}$ holds (i.e., $P$ is totally correct w.r.t.\ $\varphi_P$ and $\psi_P$).
\end{theorem}
Theorem~\ref{thm:main} means that if $\triple{\varphi_P}{P}{\psi_P}$ is proved to hold (via $T_P$), then
(1)	there exists an inference sequence of RI, and
(2) if $\cR_P$ is terminating, then $\ttriple{\varphi_P}{P}{\psi_P}$ can be proved to hold without using inference rules for proving total correctness.

\section{Conclusion}
\label{sec:conclusion}

In this paper, we showed that a proof tableau for partial correctness can be transformed into an inference sequence of RI, and also showed that if the corresponding LCTRS is terminating, then the inference sequence is valid and the program is totally correct w.r.t.\ the specified pre- and post-conditions.
Our result indicates that if we can prove partial correctness of a program by Hoare logic, then there exists a way to prove it by RI.
However, this does not mean that RI is better than Hoare logic. 
From the idea of the transformation, we may apply RI to the initial equation such as (A1) instead of constructing a proof tableau for a given Hoare triple.
Unfortunately, \textsf{Ctrl}~\cite{KN15lpar}, an RI tool for LCTRSs, did not succeed in automatically proving (A1) an inductive theorem of $\cRsumverif$.

Hoare logic often requires appropriate loop invariants, but once finding such invariants, we can construct a proof tableau in a deterministic way.
On the other hand, there must be several inference sequences of RI, and for automation, RI requires an appropriate strategy for the application of inference rules. 
In addition to the strategy, to apply {\Generalization} in this paper, we have to, given a constraint $\varphi$, find an appropriate formula $\psi$ such that $\varphi \implies \psi$ is valid and $\psi$ makes the later inference succeed.
In Section~\ref{subsec:overview}, we had the proof tableau $\Psumtab$ with an appropriate loop invariant, and thus, we could apply {\Generalization}, succeeding in transforming $\Psumtab$ into a valid inference sequence of RI.
However, this is not always possible.
For this reason, it is worth improving tools for RI so as to directly prove (A1) an inductive theorem of $\cRsumverif$.

It would be possible to transform a proof tableaux for \emph{total correctness}, which includes \emph{ranking functions} in loop invariants, into an inference sequence of RI.
However, it is not clear how to use ranking functions to prove termination of the corresponding LCTRS.
Recall that termination of programs is not preserved by the conversion to LCTRSs.
For this reason, there is a program such that there exists a ranking function to ensure termination of the program but the corresponding LCTRS is not terminating.
On the other hand, to prove validity of the converted inference sequence of RI, we can use techniques for proving termination of LCTRSs, which are based on techniques developed well for term rewriting.
The transformation of proof tableaux for partial correctness into inference sequences of RI enables us to use such techniques instead of finding appropriate ranking functions for all loops in given programs. 
The use of techniques to prove termination is one of the advantages of the transformation. 

As future work, we will transform some inference sequences of RI into proof tableaux of Hoare logic in order to compare RI with Hoare logic. 
For inference sequences of RI, we sometimes need a lemma equation that is helpful to use induction, but it is not easy to find an appropriate lemma equation.
For this reason, we expect the transformation between proof tableaux of Hoare logic and inference sequences of RI to help us to develop and improve a technique for lemma generation.

\paragraph{Acknowledgement}
We thank the anonymous reviewers of our first submission for their useful comments to improve this paper, and to encourage us to continue this work.



\end{document}